\documentclass{llncs}
\usepackage{amsmath}
\usepackage{amssymb}
\usepackage{algorithmicx}
\usepackage{algpseudocode}
\usepackage{graphicx}
\usepackage{gastex}

\spnewtheorem*{maintheorem}{Theorem}{\bf}{\itshape}
\spnewtheorem*{maincorollary}{Corollary}{\bf}{\itshape}

\begin{document}
\sloppy 

\newcommand{\lvec}[1]{\overset{{}_{\leftarrow}}{#1}}
\newcommand{\lrange}[1]{\overleftarrow{#1}}
\newcommand{\rrange}[1]{\overrightarrow{#1}}
\newcommand{\R}{\mathbb{R}}
\newcommand{\Q}{\mathbb{Q}}
\newcommand{\N}{\mathbb{N}}
\newcommand*{\orgets}{\overset{\mathrm{or}}{\gets}}
\newcommand*{\bor}{\mathrel{\mathrm{or}}}
\newcommand*{\band}{\mathop{\mathrm{and}}}
\newcommand*{\shl}{\mathrel{\mathrm{shl}}}
\def\append{{\sf append}}
\def\link{{\sf link}}
\def\unlink{{\sf unlink}}
\def\maxPal{{\sf maxPal}}
\def\nextPal{{\sf nextPal}}
\def\rad{{\sf rad}}
\def\refl{{\sf refl}}
\def\len{{\sf len}}

\algtext*{EndFunction}
\algtext*{EndProcedure}
\algtext*{EndFor}
\algtext*{EndWhile}
\algtext*{EndIf}
\algloopdefx{NIf}[1]{\textbf{if} #1 \textbf{then}}
\algloopdefx{NElse}[1]{\textbf{else} #1}
\algloopdefx{NElseIf}{\textbf{else if}}
\algloopdefx{NForAll}[1]{\textbf{for each} #1 \textbf{do}}
\algloopdefx{NWhile}[1]{\textbf{while} #1 \textbf{do}}
\algloopdefx{NFor}[1]{\textbf{for} #1 \textbf{do}}

\title{$\mathrm{Pal}^k$ Is Linear Recognizable Online}
\author{Dmitry Kosolobov \and Mikhail Rubinchik \and Arseny M. Shur}

\institute{Ural Federal University, Ekaterinburg, Russia\\ \email{dkosolobov@mail.ru}, \email{mikhail.rubinchik@gmail.com}, \email{arseny.shur@usu.ru}}
\maketitle

\begin{abstract}
Given a language $L$ that is online recognizable in linear time and space, we construct a linear time and space online recognition algorithm for the language $L\cdot\mathrm{Pal}$, where $\mathrm{Pal}$ is the language of all nonempty palindromes. Hence for every fixed positive~$k$, $\mathrm{Pal}^k$ is online recognizable in linear time and space. Thus we solve an open problem posed by Galil and Seiferas in~1978.

\noindent\textbf{Keywords:} linear algorithms, online algorithms, palindrome, factorization, palstar, palindrome decomposition
\end{abstract}

\section{Introduction}

In the last decades the study of palindromes constituted a notable branch in formal language theory. Recall that a string $w = a_1\cdots a_n$ is a \emph{palindrome} if it is equal to $\lvec{w} = a_n\cdots a_1$. There is a bunch of papers on palindromes in strings. Some of these papers contain the study of strings ``rich'' in palindromes (see, e.g.,~\cite{GJWZ}), some other present solutions to algorithmic problems like finding the longest prefix-palindrome \cite{Manacher} or counting distinct subpalindromes~\cite{KosolobovRubinchikShur}.

For languages constructed by means of palindromes, an efficient recognition algorithm is often not straightforward. In this paper we develop a useful tool for construction of acceptors for such languages. Before stating our results, we recall some notation and known facts.

The language of nonempty palindromes over a fixed alphabet is denoted by $\mathrm{Pal}$. Let $\mathrm{Pal_{ev}} = \{w\in \mathrm{Pal} \colon |w|\text{ is even}\}$, $\mathrm{Pal_{>1}} = \{w\in \mathrm{Pal} \colon |w|>1\}$. Given a function $f \colon \mathbb{N} \to \mathbb{N}$ and a language $L$, we say that an algorithm \emph{recognizes} $L$ \emph{in $f(n)$ time and space} if for any string $w$ of length $n$, the algorithm decides whether $w\in L$ using at most $f(n)$ time and at most $f(n)$ additional space. We say that an algorithm recognizes a given language \emph{online} if the algorithm processes the input string sequentially from left to right and decides whether to accept each prefix after reading the rightmost letter of that prefix.

It is well known that every context-free language can be recognized by relatively slow Valiant's algorithm (see~\cite{Valiant}). According to~\cite{Lee}, there are still no  examples of context-free languages that cannot be recognized in linear time on a RAM computer. Some ``palindromic'' languages were considered as candidates to such ``hard'' context-free languages.

At some point, it was conjectured that the languages $\mathrm{Pal_{ev}}^{*}$ and $\mathrm{Pal_{>1}}^{*}$, where $^*$ is a Kleene star, cannot be recognized in $O(n)$ (see~\cite[Section~6]{KnuthMorrisPratt}). But a linear algorithm for the former was given in~\cite{KnuthMorrisPratt} and for the latter in~\cite{GalilSeiferas}. The recognition of $\mathrm{Pal}^k$ appeared to be a more complicated problem. Linear algorithms for the cases $k = 1,2,3,4$ were given in~\cite{GalilSeiferas}. Their modified versions can be found in~\cite[Section~8]{Stringology}. In~\cite{GalilSeiferas} and~\cite{Stringology} it was conjectured that there exists a linear time recognition algorithm for $\mathrm{Pal}^k$ for arbitrary $k$. In this paper we present such an algorithm. Moreover, our algorithm is online. The main contribution is the following result.

\begin{maintheorem}
Suppose a given language $L$ is online recognizable in $f(n)$ time and space, for some function $f:\mathbb{N}\to\mathbb{N}$. Then the language $L\cdot \mathrm{Pal}$ can be recognized online in $f(n) + cn$ time and space for some constant $c > 0$ independent of $L$. \label{MainTheorem}
\end{maintheorem}

\begin{maincorollary}
For arbitrary $k$, $\mathrm{Pal}^k$ is online recognizable in $O(kn)$ time and space.
\end{maincorollary}

Note that the related problem of finding the minimal $k$ such that a given string belongs to $\mathrm{Pal}^k$ can be solved online in $O(n \log n)$ time \cite{FiciCo}, and it is not known whether a linear algorithm exists.

The paper is organized as follows. Section~\ref{palindromicPairs} contains necessary combinatorial properties of palindromes; similar properties were considered, e.g., in \cite{BreslauerGasieniec}. In Sect.~\ref{palindromicIterator} we describe an auxiliary data structure used in the main algorithm. An online recognition algorithm for $\mathrm{Pal}^k$ with $O(kn\log n)$ working time is given in Sect.~\ref{onlineAlgorithmForRecognition}. Finally, in Sect.~\ref{linearAlgorithm} we speed up this algorithm to obtain the main result.

\section{Basic Properties of Palindromes}\label{palindromicPairs}

A \emph{string of length $n$} over the alphabet $\Sigma$ is a map $\{1,2,\ldots,n\} \mapsto \Sigma$. The length of $w$ is denoted by $|w|$ and the empty string by $\varepsilon$. We write $w[i]$ for the $i$th letter of $w$ and $w[i..j]$ for $w[i]w[i{+}1]\ldots w[j]$. Let $w[i..i{-}1] = \varepsilon$ for any~$i$. A string $u$ is a \emph{substring} of $w$ if $u=w[i..j]$ for some $i$ and $j$. The pair $(i,j)$ is not necessarily unique; we say that $i$ specifies an \emph{occurrence} of $u$ in $w$. A string can have many occurrences in another string. A substring $w[1..j]$ (resp., $w[i..n]$) is a \emph{prefix} [resp. \emph{suffix}] of $w$. An integer $p$ is a \emph{period} of $w$ if $w[i] = w[i{+}p]$ for $i=1,\ldots,|w|{-}p$.

\begin{lemma}[{see~\cite[Chapter~8]{Lothaire}}]
Suppose $v$ is both a prefix and a suffix of a string $w$; then the number $|w|{-}|v|$ is a period of $w$.
\label{BorderImpliesPeriod}
\end{lemma}

A substring [resp. suffix, prefix] of a given string is called a \emph{subpalindrome} [resp. \emph{suffix-palindrome}, \emph{prefix-palindrome}] if it is a palindrome. We write $w = (uv)^*u$ to state that $w = (uv)^ku$ for some nonnegative integer $k$. In particular, $u = (uv)^*u$, $uvu = (uv)^*u$.

\begin{lemma}
Suppose $p$ is a period of a nonempty palindrome $w$; then there are palindromes $u$ and $v$ such that $|uv|=p$, $v\neq\varepsilon$, and $w=(uv)^*u$.\label{PalPeriodicity}
\end{lemma}

\begin{proof}
Let $uv$ be a prefix of $w$ of length $p$ such that $v\ne\varepsilon$ and $w=(uv)^*u$. Since $w=\lvec{w} = (\lvec{u}\lvec{v})^*\lvec{u}$, we see that $u=\lvec{u}$ and $v=\lvec{v}$.
\end{proof}

\begin{lemma}
Suppose $w$ is a palindrome and $u$ is its proper suffix-palindrome or prefix-palindrome; then the number $|w|{-}|u|$ is a period of $w$.\label{PalSuffix}
\end{lemma}

\begin{proof}
Let $w = vu$ for some $v$. Hence $vu = w = \lvec{w} = \lvec{u}\lvec{v} = u\lvec{v}$. It follows from Lemma~\ref{BorderImpliesPeriod} that $|v|$ is the period of $w$. The case of a prefix-palindrome is similar.
\end{proof}

\begin{lemma}
Let $u, v$ be palindromes such that $v \neq \varepsilon$ and $uv = z^k$ for some string $z$ and integer $k$; then there exist palindromes $x$ and $y$ such that $z = xy$, $y\neq\varepsilon$, $u = (xy)^*x$, and $v = (yx)^*y$.\label{FullPeriodPalPair}
\end{lemma}

\begin{proof}
The case $k=1$ is trivial. Suppose $k>1$. Consider the case $|z| \le |u|$. It follows from Lemma~\ref{PalPeriodicity} that there exist palindromes $x,y$ such that $z = xy$, $y\neq\varepsilon$, $u = (xy)^*x$. Since $z^k = (xy)^k = uv$, we have $v = (yx)^*y$. The case $|z|\le|v|$ is similar.
\end{proof}

A string is \emph{primitive} if it is not a power of a shorter string. A string is called a \emph{palindromic pair} if it is equal to a concatenation of two palindromes.

\begin{lemma}
A palindromic pair $w$ is primitive iff there exists a unique pair of palindromes $u, v$ such that $v\neq \varepsilon$ and $w = uv$. \label{PrimitivePalPair}
\end{lemma}

\begin{proof}
Let $w$ be a non-primitive palindromic pair. Suppose $w = z^k = uv$, where $z$ is a string, $k>1$, and $u, v$ are palindromes. By Lemma~\ref{FullPeriodPalPair}, we obtain palindromes $x,y$ such that $z = xy$ and $y\neq \varepsilon$. Now $w = u_1v_1 = u_2v_2$, where $u_1 = x, v_1 = y(xy)^{k-1}, u_2 = xyx, v_2 = y(xy)^{k-2}$ are palindromes.

For the converse, consider $w = u_1v_1 = u_2v_2$, where $u_1, u_2, v_1, v_2$ are palindromes and $|u_1| < |u_2| < |w|$. We claim that $w$ is not primitive. The proof is by induction on the length of $w$. For $|w| \le 2$, there is nothing to prove. Suppose $|w| > 2$. It follows from Lemmas~\ref{PalPeriodicity} and~\ref{PalSuffix} that there exist palindromes $x,y$ such that $u_2 = u_1yx = (xy)^*x$. In the same way we obtain palindromes $x',y'$ such that $v_1 = y'x'v_2 = (y'x')^*y'$. Hence $yx = y'x'$. Let $z$ be a primitive string such that $yx = z^k$ for some $k>0$. By Lemma~\ref{FullPeriodPalPair}, we obtain palindromes $\tilde{x}, \tilde{y}$ such that $x = (\tilde{x}\tilde{y})^*\tilde{x}$, $y = (\tilde{y}\tilde{x})^*\tilde{y}$, and $z = \tilde{y}\tilde{x}$. Similarly, we have palindromes $\tilde{x}', \tilde{y}'$ such that $x' = (\tilde{x}'\tilde{y}')^*\tilde{x}'$, $y' = (\tilde{y}'\tilde{x}')^*\tilde{y}'$, and $z = \tilde{y}'\tilde{x}'$. By induction hypothesis, $\tilde{x} = \tilde{x}'$ and $\tilde{y} = \tilde{y}'$. Finally, $w = (\tilde{x}\tilde{y})^{k'}$ for some $k'>1$.
\end{proof}

Denote by $p$ the minimal period of a palindrome $w$. By Lemma~\ref{PalPeriodicity}, we obtain palindromes $u, v$ such that $w = (uv)^*u$, $v\neq \varepsilon$, and $|uv| = p$. The string $uv$ is primitive. The representation $(uv)^*u$ is called \emph{canonical decomposition} of $w$. Let $w[i..j]$ be a subpalindrome of the string $w$. The number $(i + j)/2$ is the \emph{center} of $w[i..j]$. The center is integer [half-integer] if the subpalindrome has an odd [resp., even] length. For any integer $n$, $\shl(w,n)$ denotes the string $w[t{+}1..|w|]w[1..t]$, where $t=n\bmod |w|$.

\begin{lemma}[{see \cite[Chapter 8]{Lothaire}}]
A string $w$ is primitive iff for any integer $n$, the equality $\shl(w,n) = w$ implies $n\bmod |w| = 0$. \label{PrimitiveCriterion}
\end{lemma}

\begin{lemma}
Suppose $(xy)^*x$ is a canonical decomposition of $w$ and $u$ is a subpalindrome of $w$ such that $|u| \ge |xy|{-}1$; then the center of $u$ coincides with the center of some $x$ or $y$ from the decomposition. \label{LongSubpalindromes}
\end{lemma}

\begin{proof}[of Lemma~\ref{LongSubpalindromes}]
Consider $w = \alpha u\beta$. Since a palindrome without the first and the last letter is a palindrome with the same center, it suffices to consider the cases $|u| = |xy|{-}1$ and $|u| = |xy|$. Assume $|u| = |xy|{-}1$ (the other case is similar). Suppose there are strings $\eta$, $\theta$ and a letter $a$ such that $x = \eta a\theta$, $\alpha = (xy)^n\eta a$ for some $n\ge 0$, and $u = \theta y\eta$. (If the first letter of $u$ lies inside $y$, the proof is the same.) Then $x = \lvec{x} = \lvec{\theta}a\lvec{\eta}$, $u = \lvec{u} = \lvec{\eta}y\lvec{\theta}$. Further, $xy = \eta a\theta y = \shl(a\theta y\eta, -|\eta|) =  \shl(a\lvec{\eta}y\lvec{\theta}, -|\eta|)$. But $a\lvec{\eta}y\lvec{\theta} = \shl(\lvec{\theta}a\lvec{\eta}y, |\lvec{\theta}|) = \shl(xy, |\lvec{\theta}|)$. Hence $xy = \shl(xy, |\lvec{\theta}| - |\eta|)$. Since $xy$ is primitive, it follows from Lemma~\ref{PrimitiveCriterion} that $|\lvec{\theta}| = |\eta|$. Thus, $\lvec{\theta} = \eta$ and $u = \lvec{\eta}y\eta$.
\end{proof}

\begin{example}
Consider $x = aba$, $y = ababa$, and $u = abaaba$. Obviously, $xyxyx$ is a canonical decomposition of $aba\cdot ababa\cdot aba\cdot ababa\cdot aba$, $u$ is a suffix-palindrome of $xyxyx$, and $|u| = |xy|{-}2$. The center of $u$ is not equal to the center of $x$ or $y$. Therefore, the bound in Lemma~\ref{LongSubpalindromes} is optimal.
\end{example}

Now we briefly discuss the approach used in \cite{GalilSeiferas}. The algorithm of \cite{GalilSeiferas} essentially relies on the following ``cancelation'' lemma.

\begin{lemma}[{see \cite[Lemma C4]{GalilSeiferas}}]
Suppose $w$ is a palindromic pair; then there exist palindromes $x$ and $y$ such that $w = xy$ and either $x$ is the longest prefix-palindrome of $w$ or $y$ is the longest suffix-palindrome of $w$.
\end{lemma}

Unfortunately, it seems that even for the case of three palindromes, there are no similar results. Indeed, one can expect that if the string $s$ is a concatenation of three nonempty palindromes, then there are palindromes $x$, $y$, $z$ such that $s = xyz$ and at least one of the following statements holds:
\begin{enumerate}
\item $x$ is the longest proper prefix-palindrome of $s[1..|s|{-}1]$;
\item $z$ is the longest proper suffix-palindrome of $s[2..|s|]$;
\item $xy$ is the longest proper prefix that is a palindromic pair;
\item $yz$ is the longest proper suffix that is a palindromic pair.
\end{enumerate}

The following example shows that this hypothesis does not hold.

\begin{example}
Consider the following string that is a concatenation of three nonempty palindromes: (For convenience, some groups of letters are separated by spaces.)
$$
\underbrace{aba\;aba} \underbrace{b\;aba\;\mathbf{c}\;aba\;b\;aba\;aba\;aba\;aba\;b\;aba\;\mathbf{c}\;aba\;b} \underbrace{aba\;b\;aba}\enspace.
$$
It turns out that there are no palindromes $x$, $y$, $z$ such that $xyz$ is equal to this string and $x$, $y$, $z$ satisfy the hypothesis. To prove it, let us emphasize subpalindromes corresponding to the points of the hypothesis:
$$
\begin{array}{ll}
1. &\underbrace{aba\;aba\;b\;aba\;\mathbf{c}\;aba\;b\;aba\;aba}\;aba\;aba\;b\;aba\;\mathbf{c}\;aba\;b\;aba\;b\;aba, \\
2. &aba\;aba\;b\;aba\;\mathbf{c}\;aba\;b\;aba\;aba\;aba\;aba\;b\;aba\;\mathbf{c}\underbrace{aba\;b\;aba\;b\;aba}, \\
3. &\underbrace{aba}\;\underbrace{aba\;b\;aba\;\mathbf{c}\;aba\;b\;aba\;aba\;aba\;aba\;b\;aba\;\mathbf{c}\;aba\;b\;aba}\;b\;aba,\\
4. &aba\;a\underbrace{ba\;b\;aba\;\mathbf{c}\;aba\;b\;aba\;aba\;aba\;aba\;b\;aba\;\mathbf{c}\;aba\;b\;ab}\;\underbrace{ababa}\enspace.
\end{array}
$$
\end{example}

\section{Palindromic iterator}\label{palindromicIterator}

Let $w[i..j]$ be a subpalindrome of a string $w$. The number $\lfloor (j{-}i{+}1)/2 \rfloor$ is the \emph{radius} of $w[i..j]$. Let $\mathcal{C} = \{c>0\colon 2c\text{ is an integer}\}$ be the set of all possible centers for subpalindromes.
\emph{Palindromic iterator} is the data structure containing a string $text$ and supporting the following operations on it:
\begin{enumerate}
\item $\append_i(a)$ appends the letter $a$ to the end;
\item $\maxPal$ returns the center of the longest suffix-palindrome;
\item $\rad(x)$ returns the radius of the longest subpalindrome with the center $x$;
\item $\nextPal(x)$ returns the center of the longest proper suffix-palindrome of the suffix-palindrome with the center $x$.
\end{enumerate}

\begin{example}
Let $text = aabacabaa$. Then $\maxPal = 5$. Values of $\rad$ and $\nextPal$ are listed in the following table (the symbol ``$-$'' means undefined value):
$$
\begin{array}{r||c|c|c|c|c|c|c|c|c|c|c|c|c|c|c|c|c|c|c}
text & & a & & a & & b & & a & & c & & a & & b & & a & & a & \\
\hline
x& 0.5 & 1 & 1.5 & 2 & 2.5 & 3 & 3.5 & 4 & 4.5 & 5 & 5.5 & 6 & 6.5 & 7 & 7.5 & 8 & 8.5 & 9 & 9.5\\
\hline
\rad(x) & 0 & 0 & 1 & 0 & 0 & 1 & 0 & 0 & 0 & 4 & 0 & 0 & 0 & 1 & 0 & 0 & 1 & 0 & 0 \\
\hline
\nextPal(x) & - & - & - & - & - & - & - & - & - & 8.5 & - & - & - & - & - & - & 9 & 9.5 & -
\end{array}
$$
\end{example}

A \emph{fractional array} of length $n$ is an array with $n$ elements indexed by the numbers $\{x \in \mathcal{C} \colon 0 < x \le \frac{n}2\}$. Fractional arrays can be easily implemented using ordinary arrays of double size. Let $\refl(x, y) = y + (y - x)$ be the function returning the position symmetric to $x$ with respect to $y$.

\begin{proposition}
Palindromic iterator can be implemented such that $\append_i$ requires amortized $O(1)$ time and all other operations require $O(1)$ time.
\label{Iterator}
\end{proposition}

\begin{proof}
Our implementation uses a variable $s$ containing the center of the longest suffix-palindrome of $text$, and a fractional array $r$ of length $2s$ such that for each $i \in \mathcal{C}$, $0 < i \le s$, the number $r[i]$ is the radius of the longest subpalindrome centered at $i$. Obviously, $\maxPal = s$. Let us describe $\rad(x)$. If $x \le s$, $\rad(x) = r[x]$. If $x > s$, then each palindrome with the center $x$ has a counterpart with the center $\refl(x,s)$. On the other hand, $\rad(x)\le |text|{-}\lfloor x\rfloor$, implying $\rad(x) = \min\{r[\refl(x, s)], |text|{-}\lfloor x \rfloor\}$. To implement $\nextPal$ and $\append_i$, we need additional structures.

We define an array $lend[0..|text|{-}1]$ and a fractional array $nodes[\frac{1}2..|text|{+}\frac{1}2]$ to store information about maximal subpalindromes of $text$. Thus, $lend[i]$ contains centers of some maximal subpalindromes of the form $text[i{+}1..j]$. Precisely, $lend[i] = \{x\in \mathcal{C} \colon x<s \mbox{ and } \lceil x\rceil - \rad(x) = i{+}1\}$. Each center $x$ is also considered as an element of a biconnected list with the fields $x.\mathrm{next}$ and $x.\mathrm{prev}$ pointing at other centers. We call such elements nodes and store in the array $nodes$. The following \emph{invariant of palindromic iterator} holds.

\setlength{\leftskip}{8mm}
\noindent
\emph{Let $c_0<\ldots<c_k$ be the centers of all suffix-palindromes of $text$. For each $j\in\overline{0,k{-}1}$, $nodes[c_j].\mathrm{next} = c_{j{+}1}$ and $nodes[c_{j{+}1}].prev = c_j$}.

\setlength{\leftskip}{0pt}\smallskip
Clearly, $c_0 = s$, $c_k =|text|{+}\frac{1}2$. Let $\link(x)$ and $\unlink(x)$ denote the operations of linking $x$ to the end of the list and removing $x$ from the list, respectively. Obviously, $\nextPal(x) = nodes[x].\mathrm{next}$. The following pseudocode of $\append_i$ uses the three-operand $\mathbf{for}$ loop like in the C language.
\begin{algorithmic}[1]
{\algtext*{EndFunction}%
\Function{${\append_i}$}{$a$}
    \For{$(s_0 \gets s;\;s < |text| + 1;\;s \gets s + \frac{1}2)$}\label{lst:line:manacherBeg}
		\State $r[s] \gets \min(r[\refl(s, s_0)], |text| - \lfloor s\rfloor);$ \Comment{fill $r$}
		\If{$\lfloor s\rfloor + r[s] = |text| \mathrel{\mathbf{and}} text[\lceil s\rceil{-}r[s]{-}1] = a$}
			\State $r[s] \gets r[s] + 1;$ \Comment{here $s$ is the center of the longest suffix-pal.}
			\State $\mathbf{break};$
		\EndIf
        \State $lend[\lceil s\rceil{-}r[s]{-}1] \gets lend[\lceil s\rceil{-}r[s]{-}1] \cup \{s\};$ \Comment{fill $lend$}
	\EndFor
    \State $text \gets text \cdot a;$ \label{lst:line:manacherEnd}
    \State $\link(nodes[|text|])$; $\link(nodes[|text|{+}\frac{1}2]);$\Comment{adding trivial suffix-pals.}\label{lst:line:addPal1}
    \NForAll{$x \mathrel{\mathbf{in}} lend[\lceil s\rceil - \rad(s)]$} \label{lst:line:deleteOldSufBeg}
        \State $\unlink(nodes[\refl(x, s)]);$ \label{lst:line:deleteOldSufEnd}\Comment{removing invalid centers from the list}
\EndFunction}
\end{algorithmic}
The code in lines \ref{lst:line:manacherBeg}--\ref{lst:line:manacherEnd} is a version of the main loop of Manacher's algorithm \cite{Manacher}; see also~\cite[Chapter 8]{Stringology}. The array $lend$ is filled simultaneously with~$r$. Let us show that the invariant is preserved.

Suppose that a symbol is added to $text$ and the value of $s$ is updated. Denote by $S$ the set of centers $x > s$ such that the longest subpalindrome centered at $x$ has lost its status of suffix-palindrome on this iteration. Once we linked the one-letter and empty suffix-palindromes to the list, it remains to remove the elements of $S$ from it. Let $t = \lceil s\rceil - \rad(s)$. Since $text[t..|text|]$ is a palindrome, we have $lend[t] = \{\refl(x, s) \colon x \in S\}$. Thus, lines~\ref{lst:line:deleteOldSufBeg}--\ref{lst:line:deleteOldSufEnd} unlink $S$ from the list.

Since $\append_i$ links exactly two nodes to the list, any sequence of $n$ calls to $\append_i$ performs at most $2n$ unlinks in the loop~\ref{lst:line:deleteOldSufBeg}--\ref{lst:line:deleteOldSufEnd}. Further, any such sequence performs at most $2n$ iterations of the loop~\ref{lst:line:manacherBeg}--\ref{lst:line:manacherEnd} because each iteration increases $s$ by $\frac{1}2$ and $s \le |text|$. Thus, $\append_i$ works in the amortized $O(1)$ time.
\end{proof}

\begin{example}
Let $text = aabacaba$. The list of centers of suffix-palindromes contains $5,7,8,8.5$. Now we perform $\append_i(a)$ using the above implementation. We underline suffix-palindromes of the source string for convenience: $a\underline{{{a}ba}c\underline{ab\underline{a}}}a$. The centers $9$, $9.5$ are linked to the list in the line~\ref{lst:line:addPal1}. The set of centers to be removed from the list is $S = \{7,8\}$. Let $t = \lceil s\rceil - \rad (s) = 5 - 4 = 1$. Since $lend[t] = \{2,3\}$, the loop~\ref{lst:line:deleteOldSufBeg}--\ref{lst:line:deleteOldSufEnd} unlinks $S = \{\refl(i,s) \colon i\in lend[t]\}$ from the list. So, the new list contains $5,8.5,9,9.5$.
\end{example}

\section{Palindromic Engine}\label{onlineAlgorithmForRecognition}

\emph{Palindromic engine} is the data structure containing a string $text$, bit arrays~$m$ and~$res$ of length $|text|{+}1$, and supporting a procedure $\append(a, b)$ such that
\begin{enumerate}
\item $\append(a,b)$ appends the letter $a$ to $text$, sets $m[|text|]$ to $b$, and calculates $res[|text|]$;
\item $m$ is filled by $\append$ except for the bit $m[0]$ which is set initially;
\item $res[i] = 1$ iff there is $j\in\N$ such that $0 \le j < i$, $m[j] = 1$, and $text[j{+}1..i]\in \mathrm{Pal}$ (thus $res[0]$ is always zero).
\end{enumerate}

The following lemma is an immediate consequence of the third condition.

\begin{lemma}
Let $L$ be a language. Suppose that for any $i\in\overline{0,|text|}$, $m[i] = 1$ iff $text[1..i]\in L$; then for any $i\in\overline{0,|text|}$, $res[i] = 1$ iff $text[1..i] \in L\cdot \mathrm{Pal}$. \label{LPalRecognizer}
\end{lemma}

Let $f, g$ be functions of integer argument. We say that a palindromic engine works in $f(n)$ time and $g(n)$ space if any sequence of $n$ calls to $\append$ on empty engine requires at most $f(n)$ time and $g(n)$ space.

\begin{proposition}
Suppose a palindromic engine works in $f(n)$ time and space, and a language $L$ is online recognizable in $g(n)$ time and space; then the language $L\cdot\mathrm{Pal}$ is online recognizable in $f(n) + g(n) + O(n)$ time and space. \label{PalindromeEngineReduction}
\end{proposition}

\begin{proof}
Assume that in the palindromic engine $m[0] = 1$ iff $\varepsilon\in L$. We scan the input string $w$ sequentially from left to right. To process the $i$th letter of $w$, we feed it to the algorithm recognizing $L$ and then call $\append(w[i],1)$ or $\append(w[i],0)$ depending on whether $w[1..i]$ belongs to $L$ or not. Thus, by Lemma~\ref{LPalRecognizer}, $res[i] = 1$ iff $w[1..i] \in L\cdot\mathrm{Pal}$. Time and space bounds are obvious.
\end{proof}

We use the palindromic iterator in our implementation of palindromic engine. Let $\len(x)$ be the function returning the length of the longest subpalindrome with the center $x$, i.e., $\len(x) = 2\cdot\rad(x) + \lfloor x\rfloor - \lfloor x-\frac{1}2\rfloor$. The operations of bitwise ``or'', ``and'', ``shift'' are denoted by $\bor$, $\band$, $\shl$ respectively. Let $x \orgets y$ be short for $x \gets (x\bor y)$. The naive $O(n^2)$ time implementation is as follows:
\begin{algorithmic}[1]
{\algtext*{EndFunction}%
\Function{$\append$}{$a, b$}
    \State $\append_i(a);\;n\gets |text|;\;res[n] \gets 0;\;m[n] \gets b;$
    \NFor{$(x \gets \maxPal;\;x \neq n{+}\frac{1}2;\;x \gets \nextPal(x))$}
        \State $res[n] \orgets m[n{-}\len(x)];$ \Comment{loop through all suffix-palindromes}
\EndFunction}
\end{algorithmic}

To improve the naive implementation, we have to decrease the number of suffix-palindromes to loop through. This can be done using ``leading'' subpalindromes.

A nonempty string $w$ is \emph{cubic} if its minimal period $p$ is at most $|w|/3$. A subpalindrome $u=w[i..j]$ is \emph{leading} in $w$ if any period $p$ of any longer subpalindrome $w[i'..j]$ satisfies $2p>|u|$. Thus, all non-leading subpalindromes are suffixes of leading cubic subpalindromes. For example, the only cubic subpalindrome of $w = aabababa$ is $w[2..8] = abababa$, and the only non-leading subpalindrome is $w[4..8] = ababa$.

\begin{lemma}
Let $s = w[i..j]$ be a leading subpalindrome of $w$, with the canonical decomposition $(uv)^*u$, and $t=w[i'..j]$ be the longest proper suffix-palindrome of $s$ that is leading in $w$. Then $t = u$ if $s = uvu$, and $t = uvu$ otherwise.\label{NextLeadingPal}
\end{lemma}

\begin{proof}
Let $s = uvu$. By Lemma~\ref{PalSuffix}, $u$ is the longest proper suffix-palindrome of $s$. Clearly, $u$ is leading in $w$. If $s \neq uvu$, the assertion follows from Lemma~\ref{LongSubpalindromes}.
\end{proof}

\begin{lemma}
A string of length $n$ has at most $\log_{\frac{3}2} n$ leading suffix-palindromes.\label{LogLeadingPal}
\end{lemma}

\begin{proof}
Let $u$, $v$ be leading suffix-palindromes such that $|u| > |v|$. By Lemma~\ref{PalSuffix}, $|u|{-}|v|$ is a period of $u$. Let $p$ be the minimal period of $u$. Since $|v| < 2p$ and $p\le |u|{-}|v|$, we conclude $|u| > \frac{3}2|v|$, whence the result.
\end{proof}

To obtain a faster implementation of the palindromic engine, we loop through leading suffix-palindromes only. To take into account other suffix-palindromes, we gather the corresponding bits of $m$ into an additional bit array $z$ described below.

For every $i\in \overline{0,|text|}$, let $j_i$ be the maximal number $j'$ such that $text[i{+}1..j']$ is a leading subpalindrome. Since any empty subpalindrome is leading, $j_i$ is well defined. Let $p_i$ be the minimal period of $text[i{+}1..j_i]$. Denote by $d_i$ the length of the longest proper suffix-palindrome of $text[i{+}1..j_i]$ such that $text[j_i{-}d_i{+}1..j_i]$ is leading in $text$. By Lemma~\ref{NextLeadingPal}, $d_i = \min\{(j_i{-}i){-}p_i, p_i{+}((j_i{-}i) \bmod p_i)\}$. The array $z$ is maintained to support the following invariant:
$$
z[i] = m[i] \bor m[i{+}p_i] \bor \ldots \bor m[j_i{-}d_i{-}2p_i] \bor m[j_i{-}d_i{-}p_i]\mbox{ for all }i\in \overline{0,|text|}\, .
$$
\begin{proposition}
The palindromic engine can be implemented to work in $O(n\log n)$ time and $O(n)$ space.
\end{proposition}
\begin{proof}
Consider the following implementation of the function $\append$. An instance of its work is given below in Example~\ref{ex:nlogn}.
\begin{algorithmic}[1]
{\algtext*{EndFunction}%
\Function{$\append$}{$a, b$}
    \State $\append_i(a);\;n\gets |text|;\; res[n] \gets 0;\; m[n] \gets b;\; z[n] \gets b;\; d\gets 0;$ \label{lst:line:setZelem}
    \For{$(x \gets \maxPal;\;x \neq n{+}\frac{1}2;\;x \gets n{-}(d{-}1)/2)$} \Comment{for leading suf-pal} \label{lst:line:Head}
        \State $p \gets \len(x) - \len(\nextPal(x));$ \Comment{min period of processed suf-pal} \label{lst:line:minimalPeriod}
        \State $d \gets \min(p {+} (\len(x) \bmod p), \len(x){-}p);$ \Comment{length of next leading s-pal}\label{lst:line:uvuLength}
        \NIf{$3p > \len(x)$} \Comment{processed suf-pal is not cubic} \label{lst:line:resetZcond}
            \State $z[n{-}\len(x)] \gets m[n{-}\len(x)];$ \label{lst:line:resetZ}
        \State\textbf{else} $z[n{-}\len(x)] \orgets m[n{-}d{-}p];$ \Comment{processed suf-pal is cubic} \label{lst:line:recalcZ}
    		\State $res[n] \orgets z[n{-}\len(x)];$ \label{lst:line:modifyRes}
    \EndFor
\EndFunction} 
\end{algorithmic}
Let $w_0,\ldots, w_k$ be all leading suffix-palindromes of $text$ and $|w_0|>\ldots>|w_k|$. We show by induction that the values taken by $x$ are the centers of $w_0,\ldots, w_k$ (in this order). In the first iteration $x=\maxPal$ is the center of $w_0$. Let $x$ be the center of $w_i$. The minimal period $p$ of $w_i$ is calculated in line~\ref{lst:line:minimalPeriod} according to Lemmas~\ref{PalPeriodicity} and~\ref{PalSuffix}. By Lemma~\ref{NextLeadingPal}, the value assigned to $d$ in line~\ref{lst:line:uvuLength} is $|w_{i+1}|$. Thus, the third operand in line~\ref{lst:line:Head} sets $x$ to the center of $w_{i+1}$ for the next iteration.

Let $x$ and $(uv)^*u$ be, respectively, the center and the canonical decomposition of $w_i$. Denote by $w$ any suffix-palindrome such that $|w_i|\ge|w|>|w_{i+1}|$. By Lemma~\ref{LongSubpalindromes}, $w = (uv)^*u$. If the invariant of $z$ is preserved, the assignment in line~\ref{lst:line:modifyRes} is equivalent to the sequence of assignments $res[n] \orgets m[n{-}|w|]$ for all such $w$. Since $i$ runs from $0$ to $k$, finally one gets $res[n]\orgets m[n{-}|w|]$ for all suffix-palindromes $w$, thus providing that the engine works correctly. To finish the proof, let us show that our implementation preserves the invariant on-the-fly, setting the correct value of $z[n{-}|w_i|]$ in lines~\ref{lst:line:resetZ},~\ref{lst:line:recalcZ} just before it is used in line~\ref{lst:line:modifyRes}.

As in the pseudocode presented above, denote by $n$ the length of $text$ with the letter $c$ appended. For any $j\in \overline{0,n{-}1}$, the bit $z[j]$ is changed iff $text[j{+}1..n]$ is a leading suffix-palindrome. Assume that $w_i = text[j{+}1..n]$ is a leading suffix-palindrome and $x$ is its center. If $w_i$ is not cubic, line~\ref{lst:line:resetZ} gives the correct value of $z[j]$, because $n-d-p=j$. Suppose $w_i$ is cubic. Let $(uv)^*u$ be a canonical decomposition of $w_i$. Then $w' = text[i{+}1..n{-}|vu|]$ is a leading subpalindrome. Indeed, $w' = (uv)^*u$ and $|w'| \ge |uvuvu|$. For some $i'\le i$, suppose that $text[i'{+}1..n{-}|uv|]$ is a leading subpalindrome, $p$ is its minimal period, and $2p < |w'|$; then since $p \ge |uv|$, we have, by Lemmas~\ref{PalPeriodicity} and~\ref{LongSubpalindromes}, that either $2p > |w'|$ or $|uv|$ divides $p$. Hence $text[i'{+}1..n{-}|uv|] = (uv)^*u$. Thus $i' = i$ because $text[i{+}1..n]$ is leading. Since $w'$ is leading, we restore the invariant for $z[n{-}|w_i|]$ in line~\ref{lst:line:recalcZ}.

Since the number of iterations of the \textbf{for} cycle equals the number of leading suffix-palindromes of $text$, it is $O(\log n)$ by Lemma~\ref{LogLeadingPal}. This gives us the required time bound; the space bound is obvious.
\end{proof}

\begin{example} \label{ex:nlogn}
Let $text = ababab$. For $i = 0,\ldots,6$, denote by $j_i$ the maximal number $j'$ such that $text[i{+}1..j']$ is a leading subpalindrome. Let $(u_iv_i)^*u_i$ be a canonical decomposition of $text[i{+}1..j_i]$. We have $z[i]=m[i]$ for all $i$. The following table describes $text[i{+}1..j_i]$.
$$
\arraycolsep=3pt
\begin{array}{r||c|c|c|c|c|c|c|}
i & 0 & 1 & 2 & 3 & 4 & 5 & 6 \\
\hline
text[i{+}1..j_i] & ababa & babab & aba & bab & a & b & \varepsilon \\
\hline
u_i, v_i & a,b & b,a & a,b & b,a & \varepsilon, a & \varepsilon, b & \varepsilon, \varepsilon \\
\hline
\end{array}
$$
Assume that now we call $\append(a,m[7])$, using the $O(n\log n)$ implementation above (for simplicity, we suppose that the array $m$ is known in advance). In line 2, we get $text = abababa$, $res[7]=0$, $z[7]=m[7]$. The leading suffix-palindromes of $text$ are $w_0 = abababa$, $w_1 = aba$, $w_2 = a$, $w_3 = \varepsilon$. Then the \textbf{for} loop passes three iterations: 1) $p = 2$, $d = 3 = |aba|$; 2) $p = 2$, $d = 1 = |a|$; 3) $p = 1$, $d = 0 = |\varepsilon|$. On the first iteration, $z[0]\orgets m[2]$ is assigned. Thus, $z[0]$ takes care of the non-leading suffix-palindrome $ababa$. On the next two iterations, the condition in line 6 is true, so we (re)assign $z[4]\gets m[4]$ and $z[6]\gets m[6]$. Finally, we get
$$
res[7]=z[0]\bor z[4]\bor z[6]=m[0]\bor m[2]\bor m[4]\bor m[6].
$$
In order to demonstrate other features of this implementation of palindromic engine, we consider a few further calls to $\append$.

Suppose the next call is $\append(b,m[8])$, giving us $text = abababab$. This call is much alike the previous one, with three iterations of the \textbf{for} loop, corresponding to the suffix-palindromes $w_0 = bababab$, $w_1 = bab$, and $w_2 = b$, and the only non-trivial assignment $z[1]\orgets m[3]$.

Now let us call $\append(a,m[9])$. We again have three iterations, for $w_0 = text = ababababa$, $w_1 = aba$, $w_2 = a$. In the first iteration, $z[0]\orgets m[4]$ is assigned; this gives us $z[0]=m[0]\bor m[2]\bor m[4]$. Thus, $z[0]$ takes care of both non-leading suffix-palindromes, following  $w_0$.

Finally, assume that after several steps we have $text=ababababacabababab$ and now call $\append(a,m[19])$. On all omitted calls, no suffix-palindromes started in the first position of text, so the bit $z[0]$ was not in use and remained unchanged. On the current call, $z[0]$ is used again, but $w_0=text$ is a non-cubic palindrome; hence, $z[0]$ is reset to $m[0]$ in line 7.
\end{example}

\section{Linear Algorithm}
\label{linearAlgorithm}

Consider the \emph{word-RAM} model with $\beta{+}1$ bits in the machine word, where the bits are numbered starting with 0 (the least significant bit). A standard assumption is $\beta>\log|text|$. For a bit array $s[0..n]$ and integers $i_0$, $i_1$ such that $0\le i_1 - i_0 \le \beta$, we write $x\gets s[\rrange{i_0..i_1}]$ to get the number $x$ whose $j$th bit, for any $j\in\overline{0,\beta}$, equals $s[i_0{+}j]$ if $0\le i_0{+}j \le \min\{n,i_1\}$ and $0$ otherwise. Similarly, $x \gets s[\lrange{i_0..i_1}]$ defines $x$ with a $j$th bit equal to $s[i_1{-}j]$ if $\max\{0,i_0\}\le i_1{-}j \le n$ and to $0$ otherwise. We write $s[\lrange{i_0..i_1}] \gets x$ and $s[\rrange{i_0..i_1}] \gets x$ for the inverse operations. A bit array is called \emph{forward} [\emph{backward}] if each read/write operation for $s[\rrange{i_0..i_1}]$ [resp. $s[\lrange{i_0..i_1}]$] takes $O(1)$ time. Forward [backward] arrays can be implemented on arrays of machine words with the aid of bitwise shifts.

Processing a string of length $n$, we can read/write a group of $\log n$ elements of forward or backward array in a constant time. In this section we speed up palindromic engine using bitwise operations on groups of $\log n$ bits. This sort of optimization is often referred to as four Russians' trick (see~\cite{FourRussians}). Note that there is a simpler algorithm recognizing $\mathrm{Pal}^k$ in $O(kn\log n)$ time, but it cannot be sped up in this fashion.

In the sequel $n$ denotes $|text|$. As above, our palindromic engine contains a palindromic iterator; by Proposition~\ref{Iterator}, all computations inside the iterator take $O(n)$ total time, so we need no speed up for that.

In the implementation described below, the array $m[0..n]$ and the auxiliary array $z[0..n]$ are backward, while slightly extended array $res[0..n{+}\beta]$ is forward.

\subsubsection{5.1. Idea of the algorithm}\label{sectIdea}

We say that a call to $\append$ is \emph{predictable} if it retains the value of $\maxPal$ (or, in other words, extends the longest suffix-palindrome). For a predictable call, we know from symmetry which suffix-palindromes will be extended. This crucial observation allows us to fill $res[n..n{+}\beta]$ in advance so that in the next $\beta$ calls we need only few changes of $res$ provided that these calls are predictable.

Let $text=vs$ at some point, where $s$ is the longest suffix-palindrome. The number of subsequent calls preserving $\maxPal$ is at most $|v|=n{-}\len(\maxPal)$: this is the case if we add $\lvec{v}$. Consider those calls. Let $c_0<\ldots <c_k$ be the list of centers of all suffix-palindromes of $text$. Let $i\in \overline{1,k}$. After some predictable call $c_i$ can vanish from this list. Let $p_i$ be the number of predictable calls that retain $c_i$ on the list. Then $p_i = \rad(\refl(c_i, \maxPal)) - \rad(c_i)$ (in Fig.~\ref{fig:predic} $p_1 = 5{-}3 = 2$).

\unitlength=1mm
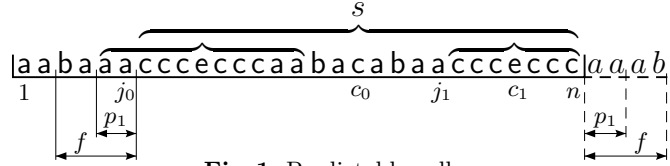
\begin{figure}[htb]
\centering
\begin{picture}(86,12)(0,4)
\gasset{Nframe=n,Nw=5,Nh=3,AHnb=0,linewidth=0.25}
\node(text)(38,12){\large\sf a\,a\,b\,a\,a\,a\,c\,c\,c\,e\,c\,c\,c\,a\,a\,b\,a\,c\,a\,b\,a\,a\,c\,c\,c\,e\,c\,c\,c}
\node(pr)(81.5,12){\large$a\,a\,a\,b$}
\node(1)(1.5,8){$1$}
\node(j0)(15,8){$j_0$}
\node(c0)(46.3,8){$c_0$}
\node(j1)(57,8){$j_1$}
\node(c1)(67.2,8){$c_1$}
\node(n)(74.5,8){$n$}
\drawline(0,13)(0,10)(76,10)(76,13)
\drawline[dash={1.5 1}{0}](76,10)(87,10)
\put(66.6,11){\makebox(0,0)[cb]{$\overbrace{\phantom{aaaaaaaaaa}}$}}
\put(25.1,11){\makebox(0,0)[cb]{$\overbrace{\phantom{aaaaaaaaaaamaaa}}$}}
\node(s)(45.9,19){\large$s$}
\put(45.9,13.5){\makebox(0,0)[cb]{$\overbrace{\phantom{aaaaaaaaaaamnaaaaaaaaaaaaaaaaaaaa}}$}}
\gasset{linewidth=0.1}
\drawline(16.4,11)(16.4,-1)
\drawline(5.7,11)(5.7,-1)
\drawline(11.1,11)(11.1,2)
\drawline[dash={1.5 1}{0}](76,11)(76,-1)
\drawline[dash={1.5 1}{0}](86.9,11)(86.9,-1)
\drawline[dash={1.5 1}{0}](81.5,11)(81.5,2)
\gasset{AHnb=1,ATnb=1,AHlength=1.2}
\node(p)(13.7,4.3){$p_1$}
\node(f)(9,1.2){$f$}
\drawline(5.7,-0.5)(16.4,-0.5)
\drawline(11.1,2.5)(16.4,2.5)
\node(p1)(78.8,4.3){$p_1$}
\node(f1)(83.5,1.2){$f$}
\drawline(76,-0.5)(86.9,-0.5)
\drawline(76,2.5)(81.5,2.5)
\end{picture}
\caption{\small Predictable calls.}
\label{fig:predic}
\end{figure}

Let $j_i=n{-}\len(c_i)$. If the operation $res[\rrange{n..n{+}\beta}] \orgets m[\lrange{j_i{-}p_i..j_i}]$ is performed for some $i\in\overline{1,k{-}1}$, we do not need to consider the suffix-palindrome with the center $c_i$ during the next $\beta$ predictable calls. Similarly, if $res[\rrange{n..n{+}\beta}] \orgets m[\lrange{j_0{-}\beta..j_0}] \bor (m[\lrange{j_k{-}p_k..j_k{-}1}] \shl 1)$ is performed, we do not consider the centers $c_0$ and $c_k$ (a shift appears because the empty suffix-palindrome is ignored). The algorithm is roughly as follows. When the assignments above are performed, each of the next $\beta$ predictable calls just adds two suffix-palindromes (one-letter and empty) and performs the corresponding assignments for them. When an unpredictable call or the $(\beta{+}1)$st predictable call occurs, we make new assignments in the current position and use array $z$ to reduce the number of suffix-palindromes to loop through. Let us consider details.

\subsubsection{5.2. Algorithm}\label{sectAlg}

We add to the engine an integer variable $f$ such that $0\le f\le \min\{\beta, n - \len(\maxPal)\}$. The value of $res[\rrange{n..n{+}f}]$ is called the \emph{prediction}. Let us describe it. The centers $c_i$ and the numbers $p_i$ are defined in Sect.~5.1. Let $\mathrm{pr} \colon \{c_0,\ldots, c_k\}\to\mathbb{N}_0$ be the mapping defined by $\mathrm{pr}(c_0)=f$ and $\mathrm{pr}(c_i)=\min\{p_i,f\}$ for $i>0$. Obviously, $\mathrm{pr}(c_i)$ is computable in $O(1)$ time. According to Sect.~5.1, the following value, called \emph{$f$-prediction}, takes care of the palindromes with the centers $c_0,\ldots,c_k$ during all the time when they are suffix-palindromes:
$$
m[\lrange{j_0{-}\mathrm{pr}(c_0)..j_0}] \bor \cdots \bor m[\lrange{j_{k{-}1}{-}\mathrm{pr}(c_{k{-}1})..j_{k{-}1}}] \bor (m[\lrange{j_k{-}\mathrm{pr}(c_k)..j_k{-}1}] \shl 1) .
$$
The prediction calculated by our algorithm will sometimes deviate from the $f$-prediction, but in a way that guarantees condition~3 of the definition of palindromic engine. Now we describe the nature of this deviation.

Let $c\in \mathcal{C}$ and $c>n{+}\frac{1}2$. Denote $c' = \refl(c, \maxPal)$. Suppose $c'>0$ and $\lceil c\rceil - \rad(c') \le n{+}1$ (see Fig.~\ref{fig:addpredic}). Let $r$ be a positive integer such that $r \le \rad(c'){+}1$ and $\lceil c\rceil - r \le n$. The values $c$ and $r$ are chosen so that after a number of predictable calls $text$ will contain a suffix-palindrome with the center $c$ and the radius $r{-}1$. Then $res[\lceil c\rceil{+}r{-}1] = 1$ if $m[\lceil c\rceil{-}r] = 1$. We call the value $g = m[\lceil c\rceil{-}r] \shl (\lceil c\rceil{+}r{-}1{-}n)]$ an \emph{additional prediction}. The assignment $res[n..n{+}f] \orgets g$ performs disjunction of the bits $res[\lceil c\rceil{+}r{-}1]$ and $m[\lceil c\rceil - r]$ (we suppose $\lceil c\rceil{+}r{-}1 \le n{+}f$). Setting this bit to $1$ is not harmful: if there will be no unpredictable calls before the position $\lceil c\rceil{+}r{-}1$, then this bit will be set to $1$ when updating the $f$-prediction on the $\lfloor c\rfloor$th iteration. Additional predictions appear as a byproduct of the linear-time implementation of the engine.

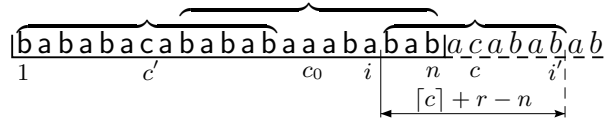
\begin{figure}[htb]
\centering
\begin{picture}(78,10)(0,4)
\gasset{Nframe=n,Nw=5,Nh=3,AHnb=0,linewidth=0.25}
\node(text)(28.8,12){\large\sf b\,a\,b\,a\,b\,a\,c\,a\,b\,a\,b\,a\,b\,a\,a\,a\,b\,a\,b\,a\,b}
\node(pr)(68.1,12){\large$a\,c\,a\,b\,a\,b\,a\,b$}
\node(1)(1.5,8){$1$}
\node(cc)(18.5,8.3){$c'$}
\node(c0)(40,7.9){$c_0$}
\node(i)(47.3,8.2){$i$}
\node(n)(56,8){$n$}
\node(c)(61.5,8){$c$}
\node(ii)(72.3,8.3){$i'$}
\drawline(0,13)(0,10)(57.5,10)(57.5,13)
\drawline[dash={1.5 1}{0}](57.5,10)(78.6,10)
\put(61.5,11.5){\makebox(0,0)[cb]{$\overbrace{\phantom{aaaaaaaaaaaaaa}}$}}
\put(17.9,11.5){\makebox(0,0)[cb]{$\overbrace{\phantom{aaaaaaaaaaaaaaaaaaaa}}$}}
\put(39.4,13.7){\makebox(0,0)[cb]{$\overbrace{\phantom{aaaamnaaaaaaaaaaaaa}}$}}
\gasset{linewidth=0.1}
\drawline(49,11)(49,2)
\drawline[dash={1.5 1}{0}](73.5,11)(73.5,2)
\gasset{AHnb=1,ATnb=1,AHlength=1.2}
\node(r)(61.3,4.3){$\lceil c\rceil+r-n$}
\drawline(49,2.5)(73.5,2.5)
\end{picture}
\caption{\small Additional prediction; $c_0=\maxPal$, $c'=\refl(c,c_0)$, $i=c-r$, $i'=c+r-1$.}
\label{fig:addpredic}
\end{figure}

We define the prediction through the \emph{main invariant of palindromic engine:} $res[\rrange{n..n{+}f}]$ equals the bitwise ``or'' of the $f$-prediction and some additional predictions. Such a definition guarantees that $res[n] = 1$ iff $m[j] = 1$ and $text[j{+}1..n] \in \mathrm{Pal}$ for some $j$, $0\le j<n$. Thus, the goal of $\append(a,b)$ is to preserve the main invariant. Our implementation of $\append(a,b)$ consists of three steps:
\begin{enumerate}
\item call $\append_i(a)$ to extend $text$ (and increment $n$); then assign $b$ to $m[n]$;
\item if $\maxPal$ remains the same and $f > 0$, decrement $f$ and perform $res[\rrange{n..n{+}f}] \orgets m[\lrange{n{-}1{-}\mathrm{pr}(n)..n{-}1}]  \bor  (m[\lrange{n{-}\mathrm{pr}(n{+}\frac{1}2)..n{-}1}] \shl 1)$;
\item otherwise, assign $f \gets \min\{\beta, n - \len(\maxPal)\}$ and recalculate the prediction $res[\rrange{n..n{+}f}]$.
\end{enumerate}
The operations of step~2 correspond to a predictable call and obviously preserve the main invariant. In the sequel we only consider step~3; step~1 is supposed to be performed: $a$ is appended to $text$, $n$ is incremented, and $m[n] = b$.

\subsubsection{5.3. Prediction recalculation} \label{sectPredRecalc}

Recall that $c_0<\ldots<c_k$ are the centers of suf\-fix-palindromes, $j_i=n{-}\len(c_i)$. First, clear the prediction: $res[\rrange{n..n{+}f}] \gets 0$. To get the $f$-prediction, it suffices to assign $res[\rrange{n..n{+}f}] \orgets m[\lrange{j_i{-}\mathrm{pr}(c_i)..j_i}]$ for $i=0,\ldots,k{-}1$ and $res[\rrange{n..n{+}f}] \orgets m[\lrange{j_k{-}\mathrm{pr}(c_k)..j_k{-}1}] \shl 1$. But our algorithm processes leading suffix-palindromes only, and the bits of $m$ that correspond to non-leading suffix-palindromes are accumulated in a certain fast accessible form in the array $z$. For simplicity, we process the empty suffix separately.

Let $i_0<\ldots<i_h$ be integers such that $c_{i_0}<\ldots<c_{i_h}$ are the centers of all leading suffix-palindromes, $r\in \overline{0,h{-}1}$ and $s = i_{r+1} - i_r - 1> 0$. Denote by $w$ the suffix-palindrome centered at $c_{i_r}$. Let $(uv)^*u$ be the canonical decomposition of $w$. It follows from Lemma~\ref{LongSubpalindromes} that $c_{i_r+1}, \ldots, c_{i_r+s}$ are the centers of $(uv)^{s+1}u, \ldots, (uv)^2u$, $c_{i_r+s+1} = c_{i_{r+1}}$ is the center of $uvu$, and $w = (uv)^{s+2}u$. Then $w$ is cubic. The converse is also true, i.e., if $w=(uv)^{s+2}u$ is a cubic suffix-palindrome, then $(uv)^{s+1}u, \ldots, (uv)^2u$ are non-leading suffix-palindromes, and $uvu$ is a leading suffix-palindrome. So, non-leading suffix-palindromes are grouped into series following cubic leading suffix-palindromes.

Recall that the palindromic iterator allows one, in $O(1)$ time, to 1) get $c_{i+1}$ from $c_i$; 2) find the minimal period of a suffix-palindrome; 3) using Lemma~\ref{NextLeadingPal}, get $c_{i_{r+1}}$ from $c_{i_r}$. The prediction recalculation involves the following steps:
\begin{enumerate}
\item accumulate some blocks of bits from $m$ into $z$ (see below); \label{stepRecalcZ}
\item for all $r\in\overline{0,h{-}1}$, assign $res[\rrange{n..n{+}f}] \orgets m[\lrange{j_{i_r}{-}\mathrm{pr}(c_{i_r})..j_{i_r}}]$;\label{stepLeading}
\item for all $r\in\overline{1,h{-}1}$, if $c_{i_r}$ is the center of a cubic suffix-palindrome and $\len(c_{i_r}) \le 2\beta$, assign $res[\rrange{n..n{+}f}] \orgets m[\lrange{j_{i_r+s}{-}\mathrm{pr}(c_{i_r+s})..j_{i_r+s}}]$ for $s = 1,2,\ldots, i_{r+1}{-}i_r{-}1$;\label{stepNonAccelerate}
\item for all $r\in\overline{0,h{-}1}$, if $c_{i_r}$ is the center of a cubic suffix-palindrome and either $\len(c_{i_r}) > 2\beta$ or $c_{i_r} = c_0$, perform the assignments of step~\ref{stepNonAccelerate} in $O(1)$ time with the aid of the array $z$.\label{stepAccelerateWithZ}
\end{enumerate}

\hyphenation{se-pa-rate-ly}
Thus, ``short'' and ``long'' non-leading suffix-palindromes are processed separately (resp., on step~\ref{stepNonAccelerate} and step~\ref{stepAccelerateWithZ}). Steps~\ref{stepRecalcZ} and~\ref{stepAccelerateWithZ} require further explanation.

\subsubsection{5.4. Content of $z$ and prediction of long suffix-palindromes}

Let $w$ be a cubic leading suffix-palindrome such that $|w| > 2\beta$ or $|w| = \len(\maxPal)$. Suppose $(uv)^*u$ is the canonical decomposition of $w$. Then $p=|uv|$ is the minimal period of $w$. Denote the centers of suffix-palindromes $w, \ldots, uvu, u$ by $c_1, c_2, \ldots, c_k$ respectively. Let us describe the behavior of those suffix-palindromes in predictable calls.

Let $t$ be the longest suffix of $text$ with the period $p$ ($t$ is not necessarily a palindrome). Then $|t| = |w| + \rad(\refl(c_k, c_1)) - \rad(c_k)$ is computable in $O(1)$ time. Since $w$ is leading and cubic, $|t| < |w| + p$. In a predictable call to $\append$, the suffix $t$ extends if $text[n] = text[n{-}p]$, and breaks otherwise. Suppose $t$ extended to $ta$. The suffix-palindromes centered at $c_2, \ldots, c_{k}$ also extended, while $w$ extends iff $|w| < |t|$. Thus, in a series of such extensions of $t$ the set of centers loses its first element during each $p$ steps. Suppose $t$ broke. Now the palindromes centered at $c_2,\ldots, c_k$ broke, while $w$ can extend provided that $w=t$.

\begin{example}\label{example7}
Let $text = baaaabaaa$. Then $\maxPal=6$; $w = aaa$ is a leading cubic suffix-palindrome; $w = (uv)^*u$ for $u = \varepsilon$ and $v = a$; $t=w$. Suffix-palindromes $aaa, aa, a, \varepsilon$ have the centers $c_1 = 8, c_2 = 8.5, c_3 = 9, c_4 = 9.5$ respectively. After the predictable call to $\append$, $text = baaaabaaaa$, $t$ is extended, and $w$ (with the center $c_1$) broke. After the second predictable call, $text = baaaabaaaab$, $t$ is broken, and only $c_2$ remains the center of a suffix-palindrome.
\end{example}

Consider the first $f$ predictable calls. Let $q$ be the maximal number such that the suffix $t$ of period $p$ extends over the first $q$ of these calls. Since $w$ is ``long'', i.e., $|w| > 2\beta$ or $w$ is the longest suffix-palindrome, and $f \le \beta$, one can be obtain $q$ in $O(1)$ time: $q = \min\{f, \rad(\refl(c_k, \maxPal)) - \rad(c_k)\}$. If $q<f$, the $(q{+}1)$st predictable call breaks the suffix of period $p$; as a result, at most one palindrome $w'=(uv)^*u$ extends to a suffix-palindrome at this moment (cf. Example~\ref{example7}). The length of $w'$ in the initial text equals $|t|{-}q$, implying $(|t|{-}q{-}|u|)\bmod p =0$. To process $w'$, we perform $res[\rrange{n..n{+}f}] \orgets m[\lrange{j{-}\mathrm{pr}(c_i)..j}]$ for $j=n{-}|w'|$, $c_i=n{-}(|w'|{-}1)/2$. To process other palindromes $(uv)^*u$, we consider $z$.

Denote $j_t = n-|t|$, $j'_t = j_t + p - 1$, and $j_w = n-|w|$, see Fig.~\ref{fig:betaacc} a,b. We store the information about the series of palindromes $(uv)^*u$ in the block $z[j_t...j'_t]$ of length $p=|uv|$. For any $j\ge 0$, $i_j = j'_t - ((j + j'_t - j_w) \bmod p)$. Thus, $i_0 = j_w$, $i_1 = j_w{-}1$ if $j_w\neq j_t$, and $i_1 = j'_t$ otherwise. Hence while $j$ increases, $i_j$ cyclically shifts left inside the range $\overline{j_t,j'_t}$. We fill the block $z[j_t..j'_t]$ such that each of its bits is responsible for the whole series of suffix-palindromes with the period~$p$.
\begin{equation} \label{eq:betaacc}
\forall j\in \overline{0,\beta}\colon z[i_j] = m[i_j] \bor m[i_j{+}p] \bor \ldots \bor m[i_j{+}lp] \mbox{ for }l = \lfloor (n {-} i_j) / p \rfloor\enspace.
\end{equation}

\begin{figure}[htb]
\centering
\begin{picture}(88,18)(0,-5)
\gasset{Nframe=n,Nw=5,Nh=3,AHnb=0,linewidth=0.25}
\node(text)(42,12){\large\sf x\,x\,x\,x\,b\,a\,c\,d\,c\,a\,b\,a\,c\,d\,c\,a\,b\,a\,c\,d\,c\,a\,b\,a\,c\,d\,c}
\node(jt)(15.2,8.1){$j_t$}
\node(jw)(20.8,8.1){$j_{\!w}$}
\node(jtt)(28.9,8.3){$j'_t$}
\node(c)(50,8){$c_1$}
\node(n)(76.3,8){$n$}
\put(1,11){\makebox(0,0)[lc]{\ldots}}
\drawline(0,10)(77.5,10)(77.5,13)
\node(s)(49.6,16.5){\large$w$}
\put(49.6,11.7){\makebox(0,0)[cb]{$\overbrace{\phantom{aaaaannnaaaaaaaaaaaaaaaaaaaaaaa}}$}}
\gasset{linewidth=0.1}
\drawline(13.9,11)(13.9,-1)
\drawline(29.8,11)(29.8,-1)
\drawline(45.7,11)(45.7,-1)
\drawline(61.6,11)(61.6,-1)
\drawline(77.5,11)(77.5,-1)
\drawline(21.9,11)(21.9,2.5)
\drawline(24.5,11)(24.5,2.5)
\drawline(37.8,11)(37.8,2.5)
\drawline(40.4,11)(40.4,2.5)
\drawline(53.7,11)(53.7,2.5)
\drawline(56.3,11)(56.3,2.5)
\drawline(69.6,11)(69.6,2.5)
\drawline(72.2,11)(72.2,2.5)
\gasset{AHnb=1,ATnb=1,AHlength=1.2}
\drawline(13.9,-0.5)(29.8,-0.5)
\drawline(29.8,-0.5)(45.7,-0.5)
\drawline(45.7,-0.5)(61.6,-0.5)
\drawline(61.6,-0.5)(77.5,-0.5)
\drawline(13.9,2.8)(21.9,2.8)
\drawline(24.5,2.8)(29.8,2.8)
\drawline(29.8,2.8)(37.8,2.8)
\drawline(40.4,2.8)(45.7,2.8)
\drawline(45.7,2.8)(53.7,2.8)
\drawline(56.3,2.8)(61.6,2.8)
\drawline(61.6,2.8)(69.6,2.8)
\drawline(72.2,2.8)(77.5,2.8)
\multiput(17.9,3.1)(15.9,0){4}{\makebox(0,0)[cb]{$r_{\!w}^{\!0}$}}
\multiput(27.2,3.1)(15.9,0){4}{\makebox(0,0)[cb]{$r_{\!w}^{\!1}$}}
\multiput(21.8,-0.2)(15.9,0){4}{\makebox(0,0)[cb]{$p$}}
\put(0,0){\makebox(0,0)[rb]{\small a}}
\end{picture}
\begin{picture}(88,19)(0,-1)
\gasset{Nframe=n,Nw=5,Nh=3,AHnb=0,linewidth=0.25}
\node(text)(42,12){\large\sf x\,x\,c\,a\,b\,a\,c\,d\,c\,a\,b\,a\,c\,d\,c\,a\,b\,a\,c\,d\,c\,a\,b\,a\,c\,d\,c}
\node(jt)(10.1,8.1){$j_t$}
\node(jw)(20.8,8.1){$j_{\!w}$}
\node(jtt)(23.6,8.3){$j'_t$}
\node(c)(50,8){$c_1$}
\node(n)(76.3,8){$n$}
\put(1,11){\makebox(0,0)[lc]{\ldots}}
\put(81,11){\makebox(0,0)[lc]{\ldots}}
\drawline(0,10)(77.5,10)(77.5,13)
\drawline[dash={1.5 1}{0}](77.5,10)(88,10)
\node(s)(49.6,16.5){\large$w$}
\put(49.6,11.7){\makebox(0,0)[cb]{$\overbrace{\phantom{aaaaannnaaaaaaaaaaaaaaaaaaaaaaa}}$}}
\gasset{linewidth=0.1}
\drawline(8.7,11)(8.7,-1)
\drawline(24.6,11)(24.6,-1)
\drawline(40.5,11)(40.5,-1)
\drawline(56.4,11)(56.4,-1)
\drawline(72.3,11)(72.3,-1)
\drawline(88.2,11)(88.2,-1)
\drawline(21.9,11)(21.9,2.5)
\drawline(37.8,11)(37.8,2.5)
\drawline(53.7,11)(53.7,2.5)
\drawline(69.6,11)(69.6,2.5)
\drawline(85.5,11)(85.5,2.5)
\gasset{AHnb=1,ATnb=1,AHlength=1.2}
\drawline(8.7,-0.5)(24.6,-0.5)
\drawline(24.6,-0.5)(40.5,-0.5)
\drawline(40.5,-0.5)(56.4,-0.5)
\drawline(56.4,-0.5)(72.3,-0.5)
\drawline(72.3,-0.5)(88.2,-0.5)
\drawline(8.7,2.8)(21.9,2.8)
\drawline(21.9,2.8)(24.6,2.8)
\drawline(24.6,2.8)(37.8,2.8)
\drawline(37.8,2.8)(40.5,2.8)
\drawline(40.5,2.8)(53.7,2.8)
\drawline(53.7,2.8)(56.4,2.8)
\drawline(56.4,2.8)(69.6,2.8)
\drawline(69.6,2.8)(72.3,2.8)
\drawline(72.3,2.8)(85.5,2.8)
\drawline(85.5,2.8)(88.2,2.8)
\multiput(15.3,3.1)(15.9,0){5}{\makebox(0,0)[cb]{$r_{\!w}^{\!0}$}}
\multiput(23.3,3.1)(15.9,0){5}{\makebox(0,0)[cb]{$r_{\!w}^{\!1}$}}
\multiput(16.7,-0.2)(15.9,0){5}{\makebox(0,0)[cb]{$p$}}
\put(0,0){\makebox(0,0)[rb]{\small b}}
\end{picture}
\caption{\small Series of palindromes with a common period $p$. The cases presented are (a) $p>\beta{+}1$ $(=r_w^0 + r_w^1 = 5)$ and (b) $\beta{+}1\ge p$ $(=r_w^0 + r_w^1 = 6)$.}
\label{fig:betaacc}
\end{figure}
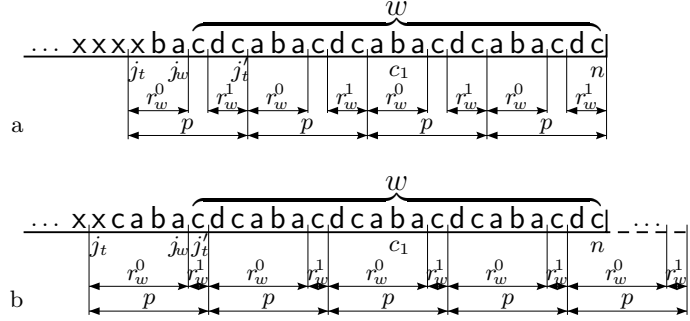

Let $r_w^0 = \min\{\beta, j_w {-} j_t\}+1$, $r_w^1 = \min\{\beta {+} 1 {-} r_w^0, j'_t {-} j_w\}$. Clearly, $r_w^0 + r_w^1 = \min\{\beta{+}1, p\}$. Hence, $i_j$ in (\ref{eq:betaacc}) runs through the ranges $[j_w{-}r_w^0{+}1..j_w]$ and $[j'_t{-}r_w^1{+}1..j'_t]$. Let $d = (1\shl (q{+}1))-1$; thus, $d$ is the bit mask consisting of $q{+}1$ ones. Suppose $\beta{+}1 < p$ (see~Fig.~\ref{fig:betaacc},a). To recalculate the prediction, it suffices to assign $res[\rrange{n..n{+}q}] \orgets d \band (z[\lrange{j_w{-}r_w^0{+}1..j_w}] \bor (z[\lrange{j'_t{-}r_w^1{+}1..j'_t}] \shl r_w^0))$. Suppose $\beta{+}1 \ge p$ (see~Fig.~\ref{fig:betaacc},b). Let $k = \lceil q / p\rceil$. To recalculate the prediction, it suffices to perform the following:
\begin{equation} \label{eq:series}
\begin{array}{l}
res[\rrange{n..n{+}q}] \orgets d \band (z[\lrange{j_t..j_w}] \bor (z[\lrange{j_w{+}1..j'_t}] \shl r_w^0)),\\
res[\rrange{n..n{+}q}] \orgets d \band ((z[\lrange{j_t..j_w}] \bor (z[\lrange{j_w{+}1..j'_t}] \shl r_w^0)) \shl p),\\[-2pt]
\ldots\\[-2pt]
res[\rrange{n..n{+}q}] \orgets d \band ((z[\lrange{j_t..j_w}] \bor (z[\lrange{j_w{+}1..j'_t}] \shl r_w^0)) \shl (kp))\enspace.
\end{array}
\end{equation}
To perform these assignments in $O(1)$ time, we use a precomputed array $g$ of length $\beta$ such that $g[i] = \sum_{j=0}^{\lfloor \beta/i\rfloor}2^{ij}$ is the bit mask containing ones separated by $i{-}1$ zeroes. Then the sequence of assignments (\ref{eq:series}) is equivalent to the operation $res[\rrange{n..n{+}q}] \orgets d \band ((z[\lrange{j_t..j_w}] \bor (z[\lrange{j_w{+}1..j'_t}] \shl r_w^0)) \cdot g[p])$.

Along with the $f$-prediction, the described method can produce additional predictions. Indeed, suppose we processed a cubic leading suffix-palindrome $w = (uv)^*u$. If $q > |v|$, the position $n{+}(|v|{+}1)/2$ is the center of the suffix-palindrome $v$ after $|v|$ predictable calls. However, the corresponding assignment $res[n{+}|v|]\orgets m[n]$ is performed much earlier: calculating the prediction in the $n$th call of $\append$, we accumulate the bit $m[n]$ in the array $z$ (see (\ref{eq:betaacc})) and then use it in updating $res[n..n{+}q]$. The assignment  $res[n{+}|v|{+}1]\orgets m[n{-}1]$ is  performed at the same moment but corresponds to the $(|v|{+}1)$st predictable call, and so on. If $q>|vuv|$, we have the same situation with the suffix-palindrome $vuv$ after $|vuv|$ calls. All these premature assignments are not necessary but bring no trouble.

\begin{lemma}
Given the array $z$, the prediction recalculation requires $O(l + \min\{2\beta, s\})$ time, where $l$ is the number of leading suffix-palindromes and $s$ is the length of the second largest leading suffix-palindrome. \label{PredictionRecalc}
\end{lemma}
\begin{proof}
The above analysis shows that each of steps~\ref{stepLeading},~\ref{stepAccelerateWithZ} takes $O(1)$ time per series of palindromes with a common period. Step~\ref{stepNonAccelerate} takes $O(\min\{2\beta, s\})$ time.
\end{proof}

\vspace*{-4mm}
\subsubsection{5.5. Recalculation of the array $z$ and the time bounds} \label{sectRecalcZ}

\begin{lemma}
Recalculation of $z$ requires $O(l + (n - n_0))$ time, where $l$ is the number of leading suffix-palindromes and $n_0$ is the length of $text$ at the moment of the previous recalculation. \label{Zrecalc}
\end{lemma}

\begin{proof}
Given a cubic leading suffix-palindrome $w$ with minimal period $p$, we set some bits in $z$ according to~(\ref{eq:betaacc}). Recall that $t$ is the longest suffix of $text$ with the period $p$, $j_t = n - |t|$, $j'_t = j_t + p - 1$, $j_w = n - |w|$, $r_w^0 = \min\{\beta, j_w-j_t\}+1$, $r_w^1 = \min\{\beta+1 - r_w^0, j'_t - j_w\}$, $r_w^0 + r_w^1=\min\{\beta{+}1, p\}$. Inside the range $\overline{j_t,j_t'}$ we have to fill the blocks $z[j_w{-}r_w^0{+}1..j_w]$ and $z[j'_t{-}r_w^1{+}1..j'_t]$. The main observation is that these blocks need only a little update after the previous recalculation.

1) \emph{Suppose $p > \beta{+}1$} (Fig.~\ref{fig:betaacc},a). If $|t| < 5p$, the assignment (\ref{eq:betaacc}) requires an $O(1)$ time both for a single bit and for a block of length $\le\beta$. Now consider the case $|t| \ge 5p$. For simplicity, the block $z[j_t..j'_t]$ is supposed to be cyclic. Then the blocks $z[j_w{-}r_w^0{+}1..j_w]$ and $z[j'_t{-}r_w^1{+}1..j'_t]$ form one segment of length $\beta$, denoted by $S$. Note that every sequence of $\beta{+}1$ calls to $\append$ contains a call that recalculates $z$. Therefore, the previous recalculation of the array $z$ filled some segment $S_1$ of length $\beta{+}1$ in $z[j_t..j'_t]$, and $S_1$ either is adjacent to $S$ from the right or overlaps $S$. Similarly, the second previous recalculation of $z$ filled some segment $S_2$ which is either adjacent to $S_1$ or overlaps it, and so on. Since $|t| \ge 5p$, all recalculations during the last $2p$ iterations processed cubic suffix-palindromes with the period $p$. In these recalculations, all positions in $S$ were filled (see~Fig.~\ref{fig:fiveperiod}). Thus, it suffices to perform $z[\lrange{j_w{-}r_w^0{+}1..j_w}] \orgets m[\lrange{j_w{-}r_w^0{+}1{+}kp..j_w{+}kp}]$, $z[\lrange{j'_t{-}r_w^1{+}1..j'_t}] \orgets m[\lrange{j'_t{-}r_w^1{+}1{+}kp..j'_t{+}kp}]$ for $k = \lfloor |t|/p \rfloor$ and $k = \lfloor |t|/p \rfloor - 1$, getting a $O(1)$ time bound again. Thus, the total time for all periods is $O(l)$.

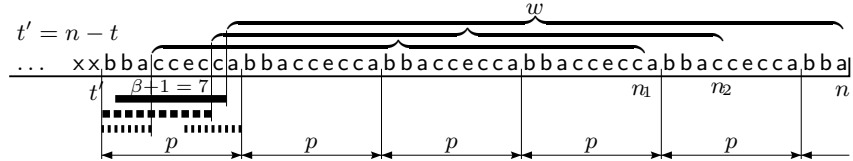
\begin{figure}[htb]
\centering
\begin{picture}(112,18)(0,0)
\gasset{Nframe=n,Nw=5,Nh=3,AHnb=0,linewidth=0.25}
\node(text)(60,12){\sf x\,x\,b\,b\,a\,c\,c\,e\,c\,c\,a\,b\,b\,a\,c\,c\,e\,c\,c\,a\,b\,b\,a\,c\,c\,e\,c\,c\,a\,b\,b\,a\,c\,c\,e\,c\,c\,a\,b\,b\,a\,c\,c\,e\,c\,c\,a\,b\,b\,a}
\node(t)(11.7,8.1){\footnotesize$t'$}
\node(n1)(84.2,8.1){\footnotesize$n_{\!1}$}
\node(n2)(94.8,8.3){\footnotesize$n_{\!2}$}
\node(n)(110.8,8){\footnotesize$n$}
\put(1,11){\makebox(0,0)[lc]{\ldots}}
\drawline(0,10)(111.7,10)(111.7,12.5)
\node(s)(69.9,19){$w$}
\put(51.7,11.3){\makebox(0,0)[cb]{$\overbrace{\phantom{aaaaannmaaaaaaaaaaaaaaaaaaaaaaaaaaaaa}}$}}
\put(60.9,12.9){\makebox(0,0)[cb]{$\overbrace{\phantom{aaaaannmnnnaaaaaaaaaaaaaaaaaaaaaaaaaaa}}$}}
\put(69.9,14.5){\makebox(0,0)[cb]{$\overbrace{\phantom{aaaaannnaaaaaaaaaaaaaaaaaaaaaaaaaaaaaaaaaaaaaaa}}$}}
\gasset{linewidth=0.1}
\drawline(12.3,11)(12.3,-1)
\drawline(30.9,11)(30.9,-1)
\drawline(49.5,11)(49.5,-1)
\drawline(68.1,11)(68.1,-1)
\drawline(86.7,11)(86.7,-1)
\drawline(105.3,11)(105.3,-1)
\drawline(28.9,16.1)(28.9,6)
\drawline(26.9,14.5)(26.9,4)
\drawline(18.9,12.9)(18.9,2)
\put(21.3,7.5){\makebox(0,0)[cb]{\scriptsize$\beta{+}1=7$}}
\put(1,14.5){\makebox(0,0)[lb]{\footnotesize$t'=n-t$}}
\drawline[linewidth=1](14.1,7)(28.9,7)
\drawline[linewidth=1,dash={1 0.5}{0}](12.3,5)(26.9,5)
\drawline[linewidth=1,dash={0.4 0.5}{0}](12.3,3)(18.9,3)
\drawline[linewidth=1,dash={0.4 0.5}{0}](22.9,3)(30.9,3)
\gasset{AHnb=1,ATnb=1,AHlength=1.2}
\drawline(12.3,-0.5)(30.9,-0.5)
\drawline(30.9,-0.5)(49.5,-0.5)
\drawline(49.5,-0.5)(68.1,-0.5)
\drawline(68.1,-0.5)(86.7,-0.5)
\drawline(86.7,-0.5)(105.3,-0.5)
\drawline[AHnb=0](105.3,-0.5)(111.7,-0.5)
\multiput(21.6,-0.2)(18.6,0){5}{\makebox(0,0)[cb]{$p$}}
\end{picture}
\caption{\small A suffix of text at the moment of recalculation of $z$. The points $n_1$ and $n_2$ of some (not necessarily last!) previous recalculations are marked; the correspondent recalculated segments of $z[j_t..j'_t]$ are shown.}
\label{fig:fiveperiod}
\end{figure}

2) \emph{Suppose $p \le \beta{+}1$} (Fig.~\ref{fig:betaacc},b). Then we must fill the whole range $z[j_t..j'_t]$. This case is similar to the above one but takes more than $O(1)$ time. We store the value $n_0$ in a variable inside the engine. Note that $n - n_0 \le \beta{+}1$. Let $t_0 = |t| - (n{-}n_0)$. If $t_0 \ge 4p$, it follows, as above, that $z[j_t..j'_t]$ contains a lot of necessary values and to fix $z[j_t..j'_t]$, we perform $z[j_t..j'_t] \orgets m[j_t{+}kp..j'_t{+}kp]$ for every integer $k$ such that $\lfloor t_0/p\rfloor \le k \le \lfloor |t|/p \rfloor$. If $t_0 < 4p$, we immediately perform $z[j_t..j'_t] \gets m[j_t..j'_t] \bor m[j_t{+}p..j'_t{+}p] \bor \cdots \bor m[j_t{+}kp..j'_t{+}kp]$ for $k = \lfloor |t|/p\rfloor$. Thus, the recalculation requires $O((n - n_0) / p)$ time. Summing up these bounds for all $p$, we get $O(n{-}n_0)$, because the values of $p$ are majorized by a geometric sequence.

Summing up the bounds for the cases 1), 2) finishes the proof.
\end{proof}

\begin{lemma}
After an unpredictable call to $\append$, $k$ successive predictable calls require $O(k)$ time in total. \label{PredictableTime}
\end{lemma}
\begin{proof}
A predictable call without recalculation takes $O(1)$ time. The number of recalculations during these $k$ calls is $\lfloor k/\beta\rfloor$. Since the number of leading suffix-palindromes is $O(\log n)$ by Lemma~\ref{LogLeadingPal}, it follows from Lemmas~\ref{PredictionRecalc},~\ref{Zrecalc} that the recalculation takes $O(\log n + \min\{2\beta, O(n)\}) + O(\log n + O(\beta)) = O(\beta)$ time, whence the result.
\end{proof}

\begin{lemma}
An unpredictable call requires $O(\maxPal{-}\maxPal_0 + n{-}n_0)$ time, where $\maxPal_0$ is the center of the longest suffix-palindrome and $n_0$ is the length of $text$ at the moment of the previous unpredictable call.\label{UnpredictableTime}
\end{lemma}
\begin{proof}
 Assume that $text=wr$ just before the current unpredictable call to $\append$, where $r$ is the longest suffix-palindrome. Note that $r$ has the center $\maxPal_0$. After this call one has $text = wrc = w't$, where $t$ is the longest suffix-palindrome. If $|t|=1$, the call takes $O(1)$ time. Suppose $t = cuc$ for some palindrome $u$. By Lemma~\ref{PalSuffix}, the number $p = |r| - |u|$ is a period of $r$. By Lemmas~\ref{PredictionRecalc} and~\ref{Zrecalc}, the prediction recalculation takes $O(l + \min\{2\beta, s\}) + O(l + (n - n_0))$ time, where $l$ is the number of leading suffix-palindromes and $s$ is the length of the second longest leading suffix-palindrome. Since $l \le s$ and $p = 2(\maxPal{-}\maxPal_0)$, it suffices to prove that $s = O(p)$. If $|u|\le \frac{2}3|r|$, then $s < |u|\le 2p$. On the other hand, if $|u|>\frac{2}3|r|$ then $p$ is a period of $t$ by~Lemma~\ref{LongSubpalindromes}. Hence $s<2p$ by the definition of a leading palindrome.
\end{proof}


\begin{proposition}
The palindromic engine can be implemented to work in $O(n)$ time and space.
\label{LinearPalindromicEngine}
\end{proposition}

\begin{proof}
The correctness of the implementation described in Sect.~5.2,~5.3 was proved in Sect.~5.2--5.4. It remains to prove the time bound. Consider the sequence of $n$ calls to $\append$. Let $n_1< n_2< \ldots < n_k$ be the numbers of all unpredictable calls to $\append$ and $\maxPal_1  < \maxPal_2 < \ldots < \maxPal_k$ be the centers of the longest suffix-palindromes just before each of these calls. By Lemma~\ref{UnpredictableTime}, all these calls require $O(1 + (\maxPal_2{-}\maxPal_1) + (n_2{-}n_1) + (\maxPal_3{-}\maxPal_2) + (n_3{-}n_2) + \ldots + (\maxPal_k{-}\maxPal_{k{-}1}) + (n_k{-}n_{k{-}1})) = O(n)$ time. A reference to Lemma~\ref{PredictableTime} ends the proof.
\end{proof}

Proposition~\ref{LinearPalindromicEngine} together with Proposition~\ref{PalindromeEngineReduction} implies the main theorem.

\section{Conclusion}

In the RAM model considered in this paper all operations are supposed to be constant-time. This is the so called \emph{unit-cost RAM}. Our algorithm heavily relies on multiplication and modulo operations, and we do not know whether it can be modified to use only addition, subtraction, and bitwise operations.

It was conjectured that there exists a context-free language that can not be recognized in linear time by a unit-cost RAM machine. This paper shows that a popular idea to use palindromes in the construction of such a language is quite likely to fail. For some discussion on this problem, see~\cite{Lee}.

\bibliographystyle{splncs}

\end{document}